\documentclass{llncs}

\usepackage{amsmath}

\usepackage{verbatim}
\usepackage{floatrow}
\usepackage{amssymb}
\usepackage{amsmath}
\usepackage{color} 
\usepackage{fullpage}
\usepackage{float}
\floatstyle{boxed} 
\restylefloat{figure}
\usepackage{cite}
\usepackage{subfig}
\usepackage{tabularx}
\usepackage{nameref}

\usepackage{color} 

\usepackage{algorithm}
\usepackage{algpseudocode}
\usepackage[titletoc]{appendix}

\ifx\pdftexversion\undefined
  \usepackage[dvips]{graphicx}
\else
  \usepackage[pdftex]{graphicx}
\fi

\newcommand{\punt}[1]{}
\newcommand{\cmnt}[1]{}

\newtheorem{observation}[theorem]{Observation}

\newcommand{\secref}[1]{Section~\ref{sec:#1}}

\newcommand{\thmref}[1]{Theorem~\ref{thm:#1}}
\newcommand{\lemref}[1]{Lemma~\ref{lem:#1}}
\newcommand{\corref}[1]{Corollary~\ref{cor:#1}}

\newcommand{\defref}[1]{Definition~\ref{def:#1}}
\newcommand{\eqnref}[1]{Eqn(\ref{eq:#1})}

\newcounter{linenumber}

\def\H{\ensuremath{\mathcal{H}}}

\newcommand{\remove}[1]{}

\newcommand{\Wset}{\textit{Wset}}
\newcommand{\Rset}{\textit{Rset}}

\newcommand{\id}[1]{\mbox{\textit{#1}}}
\newcommand{\res}[1]{\mbox{\textbf{#1}}}

\newcommand{\ignore}[1]{}

\newcommand{\op} {operation}
\newcommand{\termop} {terminal operation}

\newcommand{\nonin} {non-interference}
\newcommand{\perm} {permisiveness}

\newcommand{\lo} {\textit{LO}}
\newcommand{\clo} {\textit{CLO}}

\newcommand{\svsls} {\textit{CLO}}

\newcommand{\comm}{\textit{committed}}
\newcommand{\aborted}{\textit{aborted}}
\newcommand{\txns}{\textit{txns}}

\newcommand{\permfn}[1] {\textit{Perm}(#1)}
\newcommand{\nifn}[1] {\textit{NI}(#1)}
\newcommand{\prevA}[2] {\textit{IncAbort}(#1, #2)}

\newcommand{\evts}[1] {evts(#1)}
\newcommand{\ssch} {sub-history}

\newcommand{\lastw} {lastWrite}
\newcommand{\lwrite}[2] {#2.lastWrite(#1)}

\newcommand{\cg}[1] {CG(#1)}

\newcommand{\valid} {valid}
\newcommand{\svvalid} {legal}
\newcommand{\svvalidity} {legality}
\newcommand{\legal} {legal}
\newcommand{\legality} {legality}

\newcommand{\subs}[2]  {\textit{subC}(#2,#1)}
\newcommand{\sub}[2]  {\textit{subC}(#2,#1)}

\newcommand{\shist}[2]  {#2.subhist(#1)}

\newcommand{\applfn}[2]  {applicable(#1)}
\newcommand{\appl}  {applicable}

\newcommand{\tobj} {t-object}
\newcommand{\gchist} {\textit{gComHist}}
\newcommand{\lchist} {\textit{lComHist}}
\newcommand{\thist} {\textit{tHist}}
\newcommand{\gchlock} {\textit{gLock}}
\newcommand{\gseqn}{\textit{gSeqNum}}
\newcommand{\cseq} {\textit{commitSeqNum}}
\newcommand{\rseq} {\textit{readSeqNum}}
\newcommand{\lseq} {\textit{lseq}}
\newcommand{\tryc} {\textit{tryC}}
\newcommand{\trya}{\textit{tryA}}
\newcommand{\rvar} {\textit{readVar}}
\newcommand{\rop} {\textit{rop}}
\newcommand{\wvar} {\textit{writeVar}}
\newcommand{\wop} {\textit{wop}}
\newcommand{\cvar} {\textit{comVar}}
\newcommand{\cop} {\textit{cop}}

\newcommand{\co} {CO}
\newcommand{\opq} {opaque}
\newcommand{\opty} {opacity}

\newcommand{\coop} {co-opaque}
\newcommand{\coopty} {co-opacity}

\newcommand{\sgt} {SGT}
\newcommand{\memop} {memory operation}

\newcommand{\liveset} {liveSet}
\newcommand{\primary} {primary}
\newcommand{\glset} {gIncSet}
\newcommand{\waits} {waitSet}
\newcommand{\dset} {dependSet}
\newcommand{\sstate} {system-state}
\newcommand{\obsolete} {obsolete}

\begin{document}



\bibliographystyle{abbrv}
%

\title{Non-Interference and Local Correctness in Transactional Memory}

\author{
Petr Kuznetsov\inst{1} \and
Sathya Peri\inst{2}
}

\institute{
T\'el\'ecom ParisTech\\ petr.kuznetsov@telecom-paristech.fr \and
IIT Patna\\ sathya@iitp.ac.in
}

\remove{
\authorinfo{Petr Kuznetsov}
           {Telekom Innovation Lab, TU Berlin}
           {petr.kuznetsov@tu-berlin.de}
\authorinfo{Sathya Peri}
           {IIT Patna}
           {sathya@iitp.ac.in}
}

\date{}
\maketitle

\begin{abstract}
Transactional memory promises to make concurrent programming tractable and efficient 
by allowing the user to assemble sequences of actions in atomic \emph{transactions}
with all-or-nothing semantics.
It is believed that, by its very virtue, 
transactional memory
must ensure that 
all \emph{committed} transactions constitute a serial execution respecting
the real-time order. 
In contrast, aborted or incomplete transactions should not ``take effect.'' 
But what does ``not taking effect'' mean exactly? 

It seems natural to expect that aborted or incomplete transactions do
not appear in the global serial execution, and, thus, no committed
transaction can be affected by them.
We investigate another, less obvious, feature of ``not taking effect'' called 
\emph{non-interference}: aborted or incomplete transactions should not force
any other transaction to abort.
In the strongest form of non-interference that we explore in this paper, by removing a subset of
aborted or incomplete transactions from the history, we should not be able to 
turn an aborted transaction into a committed one without violating the
correctness criterion.     

We show that non-interference is, in a strict sense, not
\emph{implementable} with respect to 
the popular criterion of opacity that requires \emph{all} transactions (be they committed, aborted or incomplete) 
to witness the same
global serial execution. 
In  contrast, when we only require \emph{local}
correctness, non-interference is
implementable. 
Informally, a correctness criterion is local if it only 
requires that every transaction can be serialized along with (a subset
of) the transactions committed before its last event 
(aborted or incomplete transactions ignored). 
We give a few examples of local correctness properties, including the
recently proposed criterion of virtual world consistency, and present
a simple though efficient
implementation that satisfies non-interference and \emph{local opacity}.
\end{abstract}

\section{Introduction}
\label{sec:intro}

Transactional memory (TM) 
promises to make concurrent programming efficient and tractable. 
The programmer simply represents a sequence of instructions that
should appear atomic as a speculative \emph{transaction} that may
either \emph{commit} or \emph{abort}.  
It is usually expected that a TM \emph{serializes} all committed
transactions, 
i.e., makes them appear as in some sequential execution.
An implication of this requirement 
is that no committed transaction can read values written by 
a transaction that is aborted or might abort in the future. 
Intuitively, this is a desirable property because it does not allow a write performed within a
transaction to get ``visible'' as long as there is a chance for 
the transaction to abort.
   
But is this all we can do if we do not want aborted or incomplete
transactions to ``take effect''? We observe that there is a more
subtle side of the ``taking effect'' phenomenon that is usually not
taken into consideration. 
An incomplete or aborted transaction  may cause another
transaction to abort. 
Suppose we have an execution in which an aborted transaction $T$
cannot be committed without violating correctness of the execution, 
but if we remove some incomplete or aborted transactions, then $T$ can
be 
committed. 
This property, originally highlighted in~\cite{SatVid:2011:ICDCN, SatVid:2012:ICDCN}, is called \emph{non-interference}. 

Thus, ideally, a TM  must ``insulate''  transactions that are aborted or might
abort in the future from producing any effect, either by affecting reads of other transactions 
or by provoking forceful aborts.  

\vspace{1mm}
\noindent
\textit{Defining non-interference.}
Consider non-interference as a characteristics of an
\emph{implementation}.
A TM implementation $M$ is non-interfering if removing an
aborted or incomplete \emph{not concurrently committing} transaction from a 
\emph{history} (a sequence of events on the TM interface) of $M$ would
still result in a history in $M$. 
We observe that many existing TM implementations that 
employ \emph{commit-time} lock acquisition or version update  (e.g., RingSTM~\cite{spear+:spaa:ringstm:2008},
NOrec~\cite{norec}) 
are non-interfering in this sense. 
In contrast, some \emph{encounter-time} implementations,
such as TinySTM~\cite{tinystm}, are
not non-interfering.   

This paper rather focuses on non-interference as a characteristics of a
\emph{correctness criterion}, which results in a much stronger restriction
on implementations. 
We intend to understand whether this strong notion of
non-interference is achievable and at what cost, which we believe is a challenging theoretical question. 
For a given correctness criterion $C$, a TM implementation $M$ is
$C$-non-interfering if removing an aborted or incomplete transaction from any
history of $M$ does not allow committing another aborted transaction
while still preserving $C$.   
We observe that
$C$-non-interference produces a subset of
\emph{permissive}~\cite{Guer+:disc:2008} with
respect to $C$ histories. This is not difficult to see if we recall
that in a permissive (with respect to $C$) history,
no aborted transaction can be turned into a committed one
while still satisfying $C$.  

In particular, when we focus on \emph{opaque}
histories~\cite{GuerKap:2008:PPoPP,tm-book}, we observe that
non-interference gives a \emph{strict} subset of permissive opaque
histories. Opacity requires that all transactions (be they committed, aborted,
or incomplete) constitute a consistent sequential execution in which
every read returns the latest committed written value.
This is a strong requirement, because it expects every transaction
(even aborted or incomplete) to
witness the same sequential execution. 
Indeed, there exist permissive  opaque histories that do
not provide non-interference: some aborted transactions force other
transactions to abort.   
\begin{figure}[tbph]
  \centering
\includegraphics[scale=0.7]{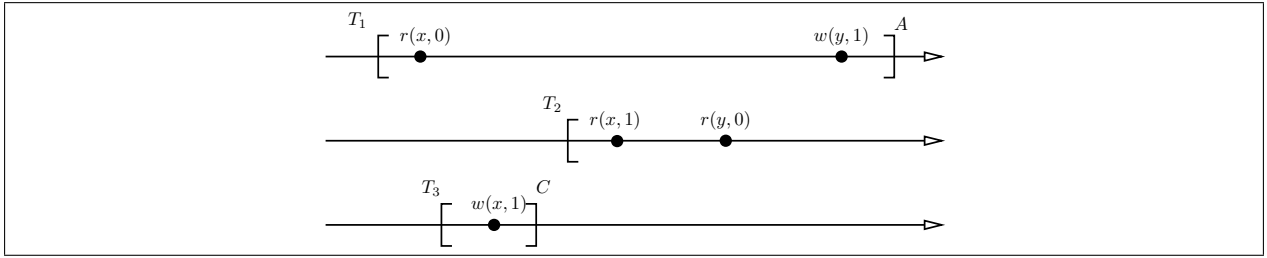}
  \caption{An opaque-permissive opaque but not opaque-non-interfering history: $T_2$ forces $T_1$
    to abort}
  \label{fig:ex1}
\end{figure}

For example, consider the history in Figure~\ref{fig:ex1}. 
Here the very fact that the incomplete operation $T_2$ read the ``new''
(written by $T_3$) value
in object $x$ and the ``old'' (initial) value in object $y$ 
prevents an updating transaction $T_1$ from committing. 
Suppose that $T_1$ commits. Then $T_2$ can only be \emph{serialized} (put in the global
sequential order) after $T_3$ and before $T_1$, while $T_1$ can only be
serialized before $T_3$. Thus, we obtain a cycle which prevents any 
serialization. 
Therefore, the history does not provide opaque-non-interference: by removing
$T_2$ we can commit $T_1$ by still allowing a correct serialization
$T_1,\;T_3$.
But the history is permissive with respect to opacity: no transaction 
aborts without a reason!

This example can be used to show that opaque-non-interference
is, in a strict sense, \emph{non-implementable}. 
Every
opaque permissive implementation that guarantees that every
transactional operation (\id{read}, \id{write}, \id{tryCommit} or \id{tryAbort})
completes if it runs in the absence of concurrency (note that it can
complete with an \emph{abort} response), may be brought to the
scenario above, where the only option for $T_1$ in its last event is \textit{abort}.      

\vspace{1mm}
\noindent
\textit{Local correctness.} 
But are there relaxed definitions of TM correctness that allow for
non-interfering implementations? Intuitively, the problem with the
history  in  Figure~\ref{fig:ex1} is that $T_2$ should be consistent
with a global order of \emph{all} transactions. But what if we only
expect every transaction $T$ to be consistent \emph{locally}, i.e., to fit to
\emph{some} serialization composed of
the transactions that committed before $T$ terminates? 
This way a transaction does not have to account for
transactions that are aborted or incomplete 
at the moment it completes and local serializations for different
transactions do not have to be mutually consistent. 

For example, the history in Figure~\ref{fig:ex1}, assuming that $T_1$
commits, is still \emph{locally} opaque:  
the local serialization of $T_2$ would simply be $T_3\cdot T_2$, while $T_1$
(assuming it commits) and $T_3$ would both be consistent with the
serialization $T_1\cdot T_3$.

In this paper, we introduce the notion of \emph{local correctness}. 
A history satisfies a local correctness property $C$ if
and only if all its ``local sub-histories'' satisfy $C$. Here a local
sub-history corresponding to $T_i$ consists of the events from all transactions that committed before the last event of $T_i$ (transactions that are incomplete or aborted at that moment are ignored) and: (1) if $T_i$ is committed then all its events; (2) if $T_i$ is aborted then all its read \op{s}.
We show that every implementation that is permissive with respect to a local correctness
criterion $C$ is also $C$-non-interfering. 

Virtual world consistency~\cite{ImbsRay:2009:SIROCCO}, that expects
the history to be strictly serializable and  
every transaction to be consistent with 
its causal past, is one example of a local correctness property.
We observe, however, that virtual world consistency may allow a transaction to
proceed even if it has no chances to commit. 
To avoid this useless work, we introduce a slightly stronger 
local criterion that we call \emph{local opacity}. As the name
suggests, a history is locally opaque if each of its local
sub-histories is opaque. 
In contrast with VWC, a locally opaque history, a transaction may only make progress if
it still has a chance to be committed. 

\vspace{1mm}
\noindent
\textit{Implementing conflict local opacity.} 
Finally, we describe a novel TM implementation that is
permissive (and, thus, non-interfering) with respect to 
\emph{conflict} local opacity (CLO). 
CLO is a restriction of local opacity that 
additionally requires each local serialization to be
consistent with the \emph{conflict order}~\cite{Papad:1979:JACM,Hadzilacos:1988:JACM}.   

Our implementation is interesting in its own right for the following
reasons.
First, it ensures non-interference, i.e., no transaction has any
effect on other transactions before committing. Second, it only
requires polynomial (in the number of concurrent transactions) local
computation for each transaction. Indeed, there are indications that,
in general, building a permissive strictly serializable TM may incur non-polynomial
time~\cite{Papad:1979:JACM}.   

The full paper is available as a technical
report~\cite{KP12:TR}. Future work includes focusing on an arguably
more practical notion of non-interference as an implementation
property, in particular, on the inherent costs of implementing non-interference.

\vspace{1mm}
\noindent
\textit{Roadmap.} 
The paper is organized as follows. We describe our system model in
Section~\ref{sec:model}. In Section~\ref{sec:ni} we formally define
the notion of $C$-non-interference, recall the definition of
permissiveness, and relate the two. In Section~\ref{sec:local}, we introduce the notion of
local correctness, show that any permissive implementation of a local
correctness criterion is also permissive,  and define the criterion of conflict local
opacity (CLO).
In Section~\ref{sec:impl} present our CLO-non-interfering implementation.
Section~\ref{sec:conc} concludes the paper with remarks on the related
work and open questions.
The appendix contains omitted definitions and proofs.

\section{Preliminaries}
\label{sec:model}

We assume a system of $n$ processes, $p_1,\ldots,p_n$ that access a
collection of \emph{objects} via atomic \emph{transactions}.
The processes are provided with  four \emph{transactional operations}: the
\textit{write}$(x,v)$ operation that updates object $x$ with value
$v$, the \textit{read}$(x)$ operation that returns a value read in
$x$, \textit{tryC}$()$ that tries to commit the transaction and
returns \textit{commit} ($c$ for short) or \textit{abort} ($a$ for
short), and \textit{\trya}$()$ that aborts the transaction and returns
$A$. The objects accessed by the read and write \op{s} are called as
\tobj{s}. For the sake of presentation simplicity, we assume that the
values written by all the transactions are unique. 

Operations \textit{write}, \textit{read} and \textit{\tryc}$()$ may
return $a$, in which case we say that the operations \emph{forcefully
abort}. Otherwise, we say that the operation has \emph{successfully}
executed.  Each operation specifies a unique transaction
identifier. A transaction $T_i$ starts with the first operation and
completes when any of its operations returns $a$ or $c$. 
Abort and commit \op{s} are called \emph{\termop{s}}. 
For a transaction $T_k$, we denote all its  read \op{s} as $Rset(T_k)$
and write \op{s} $Wset(T_k)$. Collectively, we denote all the \op{s}
of a  transaction $T_i$ as $\evts{T_k}$. 

\vspace{1mm}
\noindent
\textit{Histories.} 
A \emph{history} is a sequence of \emph{events}, i.e., a sequence of
invocation-response pairs of transactional operations. The collection
of events is denoted as $\evts{H}$. For simplicity, we only consider
\emph{sequential} histories here: the invocation of each transactional
operation is immediately followed by a matching response. Therefore,
we treat each transactional operation as one atomic event, and let
$<_H$ denote the total order on the transactional operations incurred
by $H$. With this assumption the only relevant events of a transaction
$T_k$ are of the types: $r_k(x,v)$, $r_k(x,A)$, $w_k(x, v)$, $w_k(x,
v,A)$, $\tryc_k(C)$ (or $c_k$ for short), $\tryc_k(A)$, $\trya_k(A)$ (or $a_k$ for short). 
We identify a history
$H$ as tuple $\langle \evts{H},<_H \rangle$. 

Let $H|T$ denote the history consisting of events of $T$ in $H$, 
and $H|p_i$ denote the history consisting of events of $p_i$ in $H$. 
We only consider \emph{well-formed} histories here, i.e.,
(1) each $H|T$ consists of  a read-only prefix (consisting of read
operations only), followed by a
write-only part (consisting of write operations only), possibly \emph{completed}
with a $\tryc$ or $\trya$ operation\footnote{This restriction brings no loss of 
generality~\cite{KR:2011:OPODIS}.},  and
(2) each $H|p_i$ consists of a sequence of transactions, where no new
transaction begins before the last transaction
completes (commits or aborts). 

We assume that every history has an initial committed transaction $T_0$
that initializes all the data-objects with 0. The set of transactions
that appear in $H$ is denoted by $\txns(H)$. The set of committed
(resp., aborted) transactions in $H$ is denoted by $\comm(H)$
(resp., $\aborted(H)$). The set of \emph{incomplete} transactions
in $H$ is denoted by $\id{incomplete}(H)$
($\id{incomplete}(H)=\txns(H)-\comm(H)-\aborted(H)$). 

For a history $H$, we construct the \emph{completion} of $H$, denoted $\overline{H}$, 
by inserting $a_k$ immediately after the last event 
of every transaction $T_k\in\id{incomplete}(H)$.



\vspace{1mm}
\noindent
\textit{Transaction orders.} For two transactions $T_k,T_m \in \txns(H)$, we say that  $T_k$ \emph{precedes} $T_m$ in the \emph{real-time order} of $H$, denote $T_k\prec_H^{RT} T_m$, if $T_k$ is complete in $H$ and the last event of $T_k$ precedes the first event of $T_m$ in $H$. If neither $T_k\prec_H^{RT} T_m$ nor $T_m \prec_H^{RT} T_k$, then $T_k$ and $T_m$ \emph{overlap} in $H$. A history $H$ is \emph{t-sequential} if there are no overlapping
transactions in $H$, i.e., every two transactions are related by the real-time order.

\cmnt {
We now define two order relations
on the set of transactions that are going to be instrumental in our
further definitions: real-time order, and deferred-update order.

For $T_k,T_m \in \txns(H)$, we say that  $T_k$ \emph{precedes} $T_m$
in the \emph{deferred-update order} of $H$, denote $T_k\prec_{H}^{DU} T_m$, if
$T_k$ contains a read $r_k(x,v)$, $T_m$ is committed,
$x\in\Wset(T_m)$, and $r_k(x,v)<_H c_m$. 
}

\cmnt{
For two transactions $T_k$ and $T_m$ in $\txns(H)$, we say that \emph{$T_k$ precedes $T_m$ in conflict order}, denoted $T_k \prec_H^{CO} T_m$, if (1) w-w order: $\tryc_k(C)<_H \tryc_m(C)$ and $Wset(T_k) \cap Wset(T_m) \neq\emptyset$, (2) w-r order: $\tryc_k(C)<_H r_m(x,v)$ and $x \in Wset(T_k)$, or (3) r-w order: $r_k(x,v)<_H \tryc_m(C)$ and $x\in Wset(T_m)$.

It can be seen that for any history $H$, the real-time and conflict orders are the same for $H$ and $\overline{H}$, i.e. $\prec_H^{RT} = \prec_{\overline{H}}^{RT}$ and $\prec_H^{CO} = \prec_{\overline{H}}^{CO}$. 
}

\vspace{1mm}
\noindent
\textit{Sub-histories.} A \textit{sub-history}, $SH$ of a history
$H$ denoted as the tuple $\langle \evts{SH},$ $<_{SH}\rangle$ and is
defined as: (1) $<_{SH} \subseteq <_{H}$; (2) $\evts{SH} \subseteq
\evts{H}$; (3) If an event of a transaction $T_k \in \txns(H)$ is in $SH$ then all
the events of $T_k$ in $H$ should also be in $SH$.
(Recall that $<_H$ denotes the total order of events in $H$.) 
For a history
$H$, let $R$ be a subset of $txns(H)$, the transactions in
$H$. 
Then $\shist{R}{H}$ denotes  the \ssch{} of $H$ that is
formed  from the \op{s} in $R$.

\vspace{1mm}
\noindent
\textit{Valid and legal histories.} 
Let $H$ be a history and $r_k(x, v)$ be a read {\op} in $H$. A
successful read $r_k(x, v)$ (i.e., $v \neq A$), is said to be
\emph{\valid} if there is a transaction $T_j$ in $H$ that commits
before $r_K$ and $w_j(x, v)$ is in $\evts{T_j}$. Formally, $\langle
r_k(x, v)$  is \valid{} 
$\Rightarrow \exists T_j: (c_j <_{H} r_k(x, v)) \land (w_j(x, v) \in
\evts{T_j}) \land (v \neq A) \rangle$.  The history $H$ is \valid{} 
if all its successful read \op{s} are \valid. 

We define $r_k(x, v)$'s \textit{\lastw{}} to be the latest commit event
$c_i$ such that $c_i$ precedes $r_k(x, v)$ in $H$ and $x\in\Wset(T_i)$
($T_i$ can also be $T_0$). A successful read \op{} $r_k(x, v)$ (i.e.,
$v \neq A$), is said to be \emph{\svvalid{}} if transaction $T_i$
(which contains  $r_k$'s \lastw{}) also writes $v$ onto $x$. Formally,
$\langle r_k(x, v)$ \text{is \svvalid{}} $\Rightarrow (v \neq A) \land
(\lwrite{r_k(x, v)}{H} = c_i) \land (w_i(x,v) \in \evts{T_i})
\rangle$.  The history $H$ is \svvalid{} if all its successful read
\op{s} are \svvalid. Thus from the definitions we get that if $H$ is \svvalid{} then it is also \valid.


\vspace{1mm}
\noindent
\textit{Strict Serializability and Opacity.} 
We say that two histories $H$ and $H'$ are \emph{equivalent} if
they have the same set of events.  
Now a history $H$ is said to be \textit{opaque}
\cite{GuerKap:2008:PPoPP,tm-book} 
if $H$ is \valid{}  and there exists a
t-sequential legal history $S$ such that (1) $S$ is equivalent to
$\overline{H}$ and (2) $S$ respects $\prec_{H}^{RT}$, i.e.,
$\prec_{H}^{RT} \subset \prec_{S}^{RT}$. 
By requiring $S$ being equivalent to $\overline{H}$, opacity treats
all the incomplete transactions as aborted. 

Along the same lines, a \valid{} history $H$ is said to be
\textit{strictly serializable} if $\shist{\comm(H)}{H}$ is opaque.
Thus, unlike opacity, strict serializability does not include aborted
transactions in the global serialization order.

\section{$P$-Non-Interference}
\label{sec:ni}
A \textit{correctness criterion} is a set of histories.
In this section, we recall the notion of \perm{}~\cite{Guer+:disc:2008} and 
then we formally define \nonin{}. 
%
%
First, we define a few auxiliary notions.

For a transaction  $T_i$ in $H$, \emph{\appl{}} events of $T_i$ or $\applfn{T_i}{H}$ denotes: (1) all the events of $T_i$, if it is committed; (2) if $T_i$ is aborted then all the read \op{s} of $T_i$. Thus, if $T_i$ is an aborted transaction ending with $\tryc_i(A)$ (and not $r_i(x,A)$ for some $x$), then the final $\tryc_i(A)$ is not included in $\applfn{T_i}{H}$. 

We denote, $H^{T_i}$ as the shortest prefix of $H$ containing all the events of $T_i$ in $H$. Now for $T_i\in\id{aborted}(H)$, let $\H^{T_i,C}$ denote the set of histories constructed from $H^{T_i}$, where the last operation of $T_i$ in $H$ is replaced with (1) $r_i(x,v)$ for some value non-abort value $v$, if the last operation is $r_i(x,A)$, (2) $w_i(x,v,A)$, if the last operation is $w_i(x,v,A)$, (3) $\tryc_i(C)$, if the last operation is $\tryc_i(A)$.

If $R$ is a subset of transactions of \txns($H$), then $H_{-R}$
denotes the sub-history obtained after removing all the events of $R$
from $H$. 
Respectively, $\H^{T_i,C}_{-R}$ denotes the set of histories in
$\H^{T_i,C}$ with all the events of transaction in $R$ removed.


\cmnt{
We define a variant of $\subs{T_i}{H}$, $\sub{T_i}{H}$ is the of sub-history of $H$ consisting of events from the set of committed transactions of $H$ and the events of $T_i$ \emph{up to} the last event of $T_i$, where, in case $T_i$ is aborted in $H$, the last operation of $T_i$ in $H$ is replaced with $\tryc_i(C)$,
and in case $T_i$ is incomplete in $H$, $\tryc_i(C)$ is added right after the last event of $T_i$. Formally, $\sub{T_i}{H} = \subs{T_i}{H}^{T_i,C}$ (notice that $\sub{T_i}{H}$ is not necessarily a sub-history in $H$, since some events of $H$ are modified).

If $R$ is a subset of transactions of \txns($H$), then $H_{-R}$
denotes the sub-history obtained after removing all the events of $R$
from $H$.  Thus, $\H^{T,\textbf{C}}_{-R}$ denotes the set of histories
obtained by replacing the response of the terminal \op{} of $T$ with a
non-abort value (as shown above) in the sub-history formed by removing
all the events of $R$ from $H$. 
}


\begin{definition}
\label{def:perm}
Given a correctness criterion $P$, we say that a history $H$ is
\emph{$P$-permissive}, and we write $H \in \permfn{P}$ if: 
\begin{enumerate}
\item[(1)] $H \in P$; 
\item[(2)] $\forall T \in \aborted(H)$, $\forall H' \in \H^{T,C}$: $H' \notin P$. 
\end{enumerate}
\end{definition} 
From this definition we can see that a history $H$ is permissive
w.r.t. $P$, if no aborted transaction in $H$ can be turned into
committed, while preserving $P$. 
\cmnt{It must be noted that in the above definition, the aborted
transaction $T$ is treated as read-only. }

The notion of \nonin{} or \nifn{P} is defined in a similar manner 
as a set of histories parameterized by a property $P$. 
For a transaction $T$ in $\txns(H)$, $\prevA{T}{H}$ denotes the set of
transactions that have (1) either aborted before $T$'s \termop{} or
(2) are incomplete when $T$ aborted. Hence, for any $T$,
$\prevA{T}{H}$ is a subset of $aborted(H) \cup incomplete(H)$.

\begin{definition}
\label{def:ni}
Given a correctness criterion $P$, we say that a history $H$ is
\emph{$P$-non-interfering}, and we write $H \in \nifn{P}$ if: 
\begin{enumerate}
\item[(1)] $H \in P$; 
\item[(2)] $\forall T \in \aborted(H)$, $R \subseteq \prevA{T}{H}$, 
$\forall H' \in \H^{T,C}_{-R}$:  $H' \notin P$. 
\end{enumerate}
\end{definition} 
Informally, \nonin{} states that none of transactions that
aborted prior to or are live at the moment when $T$ aborts caused $T$
to abort: removing any subset of these transactions from the history
does not help $t$ to commit.
\cmnt{By considering $\overline{H}$ in condition (2), we treat all
incomplete transactions in $H$ as aborted. It can be seen that $S$ (in
condition (2) above) could be an empty set. }
By considering the special case $R=\emptyset$ in Definition~\ref{def:ni}, we
obtain Definition~\ref{def:perm}, and, thus:
\begin{observation}
For every correctness criterion $P$, 
$\nifn{P}\subseteq \permfn{P}$. 
\end{observation}
%
%
The example in Figure~\ref{fig:ex1} (\secref{intro}) shows 
that $\nifn{\id{opacity}}\neq \permfn{\id{opacity}}$ and, thus, 
no implementation of opacity can
satisfy \nonin. 
This motivated us to define a new correctness
criterion, a relaxation of opacity, 
which satisfies \nonin.

\section{Local correctness and non-interference}
\label{sec:local}

Intuitively, a correctness criterion is local if is enough to ensure
that, for every transaction, the corresponding \emph{local sub-history}
is correct. One feature of any local property $P$ is that any $P$-permissive
implementation is also $P$-non-interfering.  

Formally, for $T_i$ in $\txns(H)$, let $\subs{T_i}{H}$ denote
\[
\shist{\comm(H^{T_i})\cup\{\applfn{T_i}{H}\}}{H^{T_i}},
\]
i.e., the sub-history of $H^{T_i}$ consisting of the events of all committed transactions in $H^{T_i}$ and all the applicable events of $T_i$. We call it local sub-history of $T_i$ in $H$. Note that here we are considering applicable events of $T_i$. So if $T_i$ is committed, all its events are considered. But if $T_i$ is an aborted transaction ending with $\tryc(A)$ (or $r_i(x,A)$), then only its read \op{s} are considered.

\begin{definition}
\label{def:local}
A correctness criterion $P$ is local if for all histories $H$: 
\begin{description}
\item $H\in P$ if and only if , for all $T_i\in\txns(H)$,
  $\subs{T_i}{H}\in P$.
\end{description}
\end{definition} 
As we show in this section,  one example of a local property is virtual world
consistency~\cite{ImbsRay:2009:SIROCCO}.
Then we will introduce another local property that we call conflict local opacity (\clo), in the next section and describe a simple permissive \clo{} implementation.


\begin{theorem}
\label{thm:perm_ni}
For every local correctness property $P$, $\permfn{P}\subseteq \nifn{P}$.
\end{theorem}

\begin{proof}
We proceed by contradiction. Assume that $H$ is in
$\permfn{P}$ but not in $\nifn{P}$. 
More precisely, let $T_a$ be an aborted transaction in $H$, 
$R \subseteq \prevA{T_a}{H}$ and  
$\widetilde H \in \H^{T_a,C}_{-R}$, such that  $\widetilde H \in P$.

On the other hand, since $H\in \permfn{P}$, we have  
$\H^{T_a,C} \cap P = \emptyset$.
Since $P$ is local and $H\in P$, we have $\forall T_i\in\txns(P)$, 
 $\subs{T_i}{H}\in P$. 
Thus, for all transactions $T_i$ that
 committed before the last event of $T_a$, we have  
$\subs{T_i}{H}=\subs{T_i}{H^{T_a}}\in P$.    

Now we construct $\widehat H$ as $H^{T_a}$, except that the aborted
operation of $T_a$ is replaced with the last operation of $T_a$ in
$\widetilde H$.  Since $\widetilde H$ is in $P$, and $P$ is local, 
we have $\subs{T_a}{\widehat H} =
\subs{T_a}{\widetilde H}\in P$. 
For all transactions $T_i$ that committed before the last event of $T_a$ in
$\widehat H$, we have  $\subs{T_i}{\widehat H}=\subs{T_i}{H^{T_a}}\in P$.
Hence, since $P$ is local, we have $\widehat H\in P$.
But, by construction, $\widehat H\in\H^{T_a,C}$---a contradiction with
the assumption that $\H^{T_a,C} \cap P = \emptyset$.
\qed
\end{proof}

\noindent
As we observed earlier, for any correctness criterion $P$, $\nifn{P} \subseteq \permfn{P}$. Hence, \thmref{perm_ni} implies that for any local correctness
criterion $P$ $\nifn{P} = \permfn{P}$.

\subsection{Virtual world consistency}
\label{subsec:vwc}

The correctness criterion of \emph{virtual world
  consistency} (VWC)~\cite{ImbsRay:2009:SIROCCO} relaxes opacity by
 allowing aborted transactions to be only consistent with its local
 \emph{causal} past.
More precisely, we say that $T_i$ \emph{causally precedes} $T_j$ in a history
$H$, and we write $T_i\prec_H^{CP} T_j$  if one of the following
conditions hold (1) $T_i$ and $T_j$ are executed by the same process
and   $T_i\prec_H^{RT}T_j$, (2) $T_i$ commits and $T_j$  reads the
value written by $T_i$ to some object $x\in\Wset(T_i)\cap\Rset(T_j)$(recall that we assumed for simplicity that all
written values are unique),   or (3) there exists $T_k$, such that
$T_i\prec_H^{CP} T_k$  and $T_k\prec_H^{CP} T_j$.  
The set of  transactions $T_i$ such that
$T_i\prec_H^{CP}T_j$ and $T_j$ itself is called the \emph{causal past} of
$T_j$, denoted $CP(T_j)$.    

Now $H$ is in \id{VWC} if (1) $\shist{\comm}{H}$ is opaque
and (2) for every $T_i\in\txns(H)$,  $\shist{CP(T_i)}{H}$ is opaque. 
Informally, $H$ must be strictly serializable and the causal past 
of every transaction in $H$ must
constitute an opaque history. 

It is easy to see that $H\in\id{VWC}$ if and only if for all
$\subs{T_i}{H}\in\id{VWC}$. 
By \thmref{perm_ni}, any VWC-permissive implementation is
also VWC-non-interfering.  

\subsection{Conflict local opacity}
\label{subsec:clo}

As shown in~\cite{ImbsRay:2009:SIROCCO}, the $\id{VWC}$ criterion may allow a transaction to proceed if it is ``doomed'' to abort: as long as the transaction's causal past can be properly serialized, the transaction may continue if it is no more consistent with the global serial order and, thus, will have to eventually abort. We propose below a stronger local property that, intuitively, aborts a transaction as soon as it cannot be put in a global serialization order. 

\begin{definition}
\label{def:lo}
A history $H$ is said to be \emph{locally opaque} or \emph{LO},
if for each transaction $T_i$ in $H$:
$\subs{T_i}{H}$ is opaque. 
\end{definition}

It is immediate from the definition that a locally opaque history is
strictly serializable: 
simply take $T_i$ above to be the last transaction to commit  in $H$. 
The resulting $\subs{T_i}{H}$ is going to be $\shist{\comm(H)}{H}$, 
the sub-history consisting of all committed transactions in $H$.     
Also, one can easily see that local opacity is indeed a local property. 

Every opaque history is also locally opaque, but not vice versa.  
To see this, consider the history $H$ in Figure~\ref{fig:ex3} which is like the
history in  Figure~\ref{fig:ex1}, except that transaction $T_1$ is now
committed.
Notice that the history is not opaque anymore: $T_1$, $T_2$ and $T_3$
form a cycle that prevents any legal serialization.
But it is \emph{locally} opaque: 
each transaction witnesses a state which is
consistent with some legal total order on transactions committed so far:
$\subs{T_1}{H}$ is equivalent to $T_3 T_1$, 
$\subs{T_2}{H}$ is equivalent to $T_3 T_2$, 
$\subs{T_3}{H}$ is equivalent to $T_3$.

\begin{figure}[tbph]
  \centering
\includegraphics[scale=0.7]{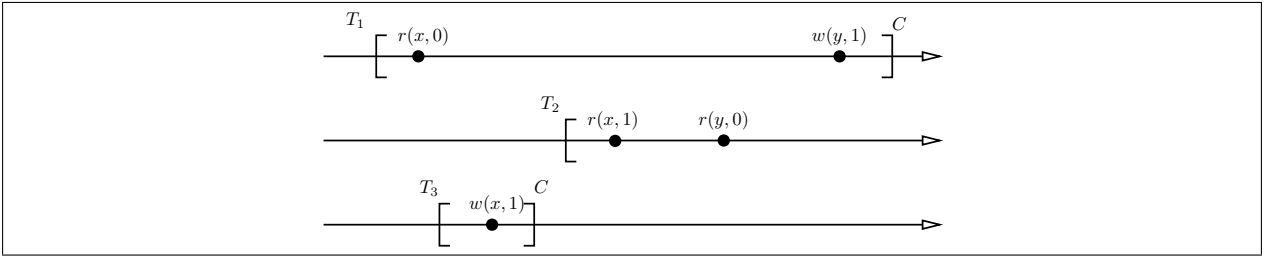}
  \caption{A locally opaque, but not opaque history (the initial value for each
object is $0$)}
  \label{fig:ex3}
\end{figure}

We denote the set of locally opaque histories by $\lo$. Finally, we
propose a restriction of local  opacity that  ensures that every local
serialization respects the \textit{conflict
  order}~\cite[Chap. 3]{WeiVoss:2002:Morg}. 
For two transactions $T_k$ and $T_m$ in $\txns(H)$, we say that \emph{$T_k$
  precedes $T_m$ in conflict order}, denoted $T_k \prec_H^{CO} T_m$, 
if (w-w order) $\tryc_k(C)<_H \tryc_m(C)$ and $Wset(T_k) \cap
Wset(T_m) \neq\emptyset$, 
(w-r order) $\tryc_k(C)<_H r_m(x,v)$, $x \in Wset(T_k)$ and $v \neq
A$, or (r-w order) 
$r_k(x,v)<_H \tryc_m(C)$, $x\in Wset(T_m)$ and $v \neq A$.
Thus, it can be seen that the conflict order is defined only on \op{s}
that have successfully executed. 
Using conflict order, we define a subclass of opacity, conflict opacity (\coopty). 
\begin{definition}
\label{def:coop1}
A history $H$ is said to be \emph{conflict opaque} or \emph{\coop} if $H$ is \valid{} 
and there exists a t-sequential legal history $S$ such that (1) $S$ is equivalent to
$\overline{H}$ and (2) $S$ respects $\prec_{H}^{RT}$ and $\prec_{H}^{CO}$. 
\end{definition}
\cmnt{Readers familiar with the databases literature may notice \coopty{} 
is analogous to the \emph{order commit conflict serializability} (OCSR)~\cite{WeiVoss:2002:Morg}.}

Now we define a ``conflict'' restriction of local opacity, \textit{conflict local opacity 
(\clo)} by replacing \opq{} with \coop{} in Definition~\ref{def:lo}.
Immediately, we derive that \coopty{} is a subset of \opty{} and \clo{} is a subset of \lo.




\section{Implementing Local Opacity}
\label{sec:impl}

In this section, we present our permissive implementation of \clo. 
By \thmref{perm_ni} it is also \clo-non-interfering. 
Our implementation is based on conflict-graph construction of \coopty{},
a popular technique borrowed from databases
(cf.~\cite[Chap.~3]{WeiVoss:2002:Morg}). 
We then describe a simple garbage-collection optimization that prevents the
memory used by the algorithm from growing without bound.

\subsection{Graph characterization of \coopty}
\label{subsec:graph}

Given a history $H$, we construct a \textit{conflict graph}, $\cg{H} =
(V,E)$ as follows:  (1) $V=\txns(H)$, the set of transactions in $H$
(2) an edge $(T_i,T_j)$ is added to $E$ whenever 
$T_i \prec_H^{RT} T_j$ or $T_i \prec_H^{CO} T_j$, i.e., whenever 
$T_i$ precedes $T_j$ in the real-time or conflict order.
\cmnt{
\begin{itemize}
\item[2.1] Real-Time edges: If $T_i$ precedes $T_j$ in $H$
\item[2.2] \co{} edges: If one of the following conditions hold (a) w-w edges: $c_i<_{H} c_j$ and $Wset(T_i) \cap Wset(T_j) \neq \emptyset$, (b) w-r edges: $c_i<_{H} r_j(x,v)$ and $x \in Wset(T_j)$, or (c) r-w edges: $r_i(x,v)<_{H} c_j$ and $x \in Wset(T_j)$. 
\end{itemize}
}

\noindent From this characterization, we get the following theorem:

\begin{theorem}
\label{thm:graph}
A \legal{} history $H$ is \coop{} iff $\cg{H}$ is acyclic. 
\end{theorem}

\subsection{The Algorithm for Implementing \clo}
\label{subsec:algo}

Our \clo{} implementation is presented in Algorithms~\ref{algo:read}, \ref{algo:write} and \ref{algo:tryc} (we omit the trivial implementation of \trya{} here).
The main idea is that the system maintains a sub-history of 
all the committed transactions. 
Whenever a live transaction $T_i$  wishes to perform an \op{} $o_i$ 
(read, write or commit), the
TM system checks to see if $o_i$ and the transactions that committed
before it,  form a cycle. If so, $o_i$ is not permitted to execute and
$T_i$ is aborted.  Otherwise, the \op{} is allowed to execute. 
Similar algorithm(s) called as serialization graph testing have been
proposed for databases (cf.~\cite[Chap.~4]{WeiVoss:2002:Morg}). Hence, we call it \textit{\sgt{}} algorithm. 


\begin{algorithm}
\caption{Read of a \tobj{} $x$ by a transaction $T_i$} \label{algo:read}
\begin{algorithmic}[1]
\algrenewcommand{\algorithmiccomment}[1]{// \textsf{#1}}

\Procedure {$\textit{read}_i$}{$x$}  
  \State \Comment{read \gchist{}}
  \State $\thist_i = \gchist$;  \Comment{create a local copy of
\gchist{}} \label{lin:ghist}
  \State \Comment{create $v$, to store a the value of $x$} 
  \State $v =$ the latest value written to $x$ in $\thist_i$; 
  \State \Comment{create $\lseq_i$, the local copy of $gseqn$} 
  \State $\lseq_i =$ the value of largest seq. no. of a transaction in $\lchist_i$; 
  \State create the \rvar{} $\rop_i(x, v, \lseq_i)$;
  \State \Comment{update $\lchist_i$}    
  \State $\lchist_i =$ merge $\lchist_i$ and $\thist_i$; append $\rop_i(x, v,
\lseq_i)$ to $\lchist_i$; \label{lin:lchist1}
  \State \Comment{check for consistency of the read \op}
  \If {$(CG(\lchist_i)$ is cyclic)} \label{acycl:check1}
    \State replace $\rop_i(x, v, \lseq_i)$ with $(\rop_i(x, A, \lseq_i)$
					in $\lchist_i)$; 
    \State return \texttt{abort};
  \EndIf    
  \State \Comment{current read is consistent; hence store it in the read set and return $v$}
  \State return $v$; 
\EndProcedure
\end{algorithmic}
\end{algorithm}

\begin{algorithm}
\caption{Write of a \tobj{} $x$ with value $v$ by a transaction $T_i$}\label{algo:write} 
\begin{algorithmic}[1]
\algrenewcommand{\algorithmiccomment}[1]{// \textsf{#1}}


\Procedure {$\textit{write}_i$}{$x, v$} 

  \If {$\textit{write}_i(x,v)$ is the first operation in $T_i$}    
	  \State \Comment{read \gchist{}}
	  \State $\lchist_i = \gchist$; 
	  \State $\lseq_i =$ the value of largest seq. no. of a
			transaction in $\lchist_i$; 
  \EndIf
\State create the \wvar{} $\wop_i(x, v,\lseq_i)$;
\State append $\wop_i(x, v, \lseq_i)$ to $\lchist_i$;
\State return $ok$; 

\EndProcedure
\end{algorithmic}
\end{algorithm}

\cmnt{

\begin{algorithm}
\caption{Insert the \rvar{} $\rop_i$ of transaction $T_i$ into $\lchist_i$ using seq num $s$}
\begin{algorithmic}[1]
\algrenewcommand{\algorithmiccomment}[1]{// \textsf{#1}}
{\small

\Procedure {$insert$}{$\rop_i(x, v, s), \lchist_i$}  
  \State Identify transaction $T_j$ whose $\cseq_j$ is $s$;
  \State Insert $\rop_i(x, v, s)$ after \cvar{} $\cop_j$ but before \cvar{} $\cop_{j+1}$ into $\lchist_i$;
  \State return; 
\EndProcedure
}
\end{algorithmic}
\end{algorithm}

\begin{algorithm}
\caption{Modify the \rvar{} $\rop_i$ of transaction $T_i$ in $\lchist_i$}
\begin{algorithmic}[1]
\algrenewcommand{\algorithmiccomment}[1]{// \textsf{#1}}

{\small

\Procedure {$modify$}{$\rop_i(x, v, s), \lchist_i$}  
  \State Identify the $\rop_i(x, v', s)$ of $T_i$ in $\lchist_i$;
  \State Replace $v'$ with $v$ in $\rop_i(x, v', s)$;
  \State return; 
\EndProcedure
}
\end{algorithmic}
\end{algorithm}

\begin{algorithm}
\caption{Constructs the conflict graph and checks for cycles in $\lchist_i$}
\begin{algorithmic}[1]
\algrenewcommand{\algorithmiccomment}[1]{// \textsf{#1}}

{\small

\Procedure {$isCyclic$}{$\lchist_i$}  
  \State Construct the conflict graph $\cg{\lchist_i}$;
  \If {$\cg{\lchist_i}$ contains a cycle}
    \State return \texttt{true};
  \Else
    \State return \texttt{false};
  \EndIf
\EndProcedure
}
\end{algorithmic}
\end{algorithm}
}

\begin{algorithm}
\caption{TryCommit \op{} by a transaction $T_i$} \label{algo:tryc}
\begin{algorithmic}[1]
\algrenewcommand{\algorithmiccomment}[1]{// \textsf{#1}}


\Procedure {$\tryc_i$} {}
  \State lock $\gchlock$;
  \State \Comment{create the next version of $gseqn$ for the current $T_i$}
  \State $\lseq_i = \gseqn + 1$;      
  \State $\thist_i = \gchist$; \Comment{create a local copy of
\gchist{}}
  \State $\lchist_i =$ merge $\lchist_i$ and $\thist_i$; \Comment{update $\lchist_i$} 
  \State \Comment{create the commit \op{} with $\lseq_i$}
  \State create the \cvar{} $\cop_i(\lseq_i)$;
  \State append $\cop_i(\lseq_i)$ to $\lchist_i$; \label{lin:lchist2}
  \If {$(CG(\lchist_i)$ is cyclic)} \label{acycl:check2}
    \State Replace $\cop_i(\lseq_i)$ with $a_i$ in $\lchist_i$;
    \State Release the lock on $\gchlock$;
    \State return \texttt{abort};
  \EndIf    

  \State $\gchist = \lchist_i$;  \label{lin:update}
  \State $\gseqn = \lseq_i $; 
  \State Release the lock on $\gchlock$;
  \State return \texttt{commit};
\EndProcedure
\end{algorithmic}
\end{algorithm}

Our \sgt{} algorithm maintains several variables. Some of them are global to
all transactions which are prefixed with the letter `g'. The remaining
variables are local. The variables are: 
\begin{itemize}

\item \textit{\gseqn}, initialized to $0$ in the start of the system:
global variable that counts the number
of transactions committed so far. 

\item $\lseq_i$: a transaction-specific variable that contains the number of
transactions currently observed committed by $T_i$. 
When a transaction $T_i$ commits, the current value of \gseqn{} is
incremented and assigned to $\lseq_i$. 

\cmnt{
\item \textit{\rseq}: when a transaction $T_i$ performs a read \op{}
$r_i(x, v)$, the algorithm associates a sequence number \res{} with
the read \op. Suppose $r_i$'s reads $x$ from $T_j$ (a committed transaction). Then, $r_i$'s \rseq{} $\lseq_j$. We denote the
read's sequence number as $r_i(x, v).\rseq$. 
}

\item \textit{\rvar}: captures a read \op{} $r_i$ performed by a transaction
$T_i$. It stores the variable $x$, the value $v$
returned by $r_i$ and the \emph{sequence number} $s$ of $r_i$,
computed as the sequence number of the committed transaction $r_i$ reads from.
We use the notation $\rop_i(x, v, s)$ to denote the read \op{} in the
local or global history.

\cmnt{
\item $\Rset_i$, initialized to $\emptyset$: the set of all \rvar{}
\op{s} performed by $T_i$.  
\cmnt{If $T_i$ performs $r_i(x,v)$ then $x, v$ and the $\rseq{}$ are
stored in this list. }
}

\item \textit{\wvar}: captures a write \op{} $w_i(x, v)$ performed by
a transaction $T_i$. It stores the variable $x$, the value written by
the write \op{} $v$ and the sequence number $s$ of $w_i$,
computed as the sequence number of the previous op{} in $T_i$ or the
sequence number of the last committed transaction preceding $T_i$ if
$w_i$ is the first operation in $T_i$.
We use the notation $\wop_i(x, v, s)$ to denote the \wvar{} \op{}.

\cmnt{
\item $\Wset_i$, initialized to $\emptyset$: the set of all \wvar{}
\op{s} performed by $T_i$. 
}

\item \textit{\cvar}: captures a commit \op{} of a transaction $T_i$. 
It stores the $\lseq_i$ of the transaction. We use the notation
$\cop_i(s)$ to denote the \cvar{} \op{} where $s$ is the $\lseq_i$ of
the transaction. 

\item \textit{\gchist}: captures the history of events of committed
transactions. It is a list of \rvar, \wvar, \cvar{} variables ordered 
by real-time execution. 
We assume that $\gchist$ also contains initial 
values for all t-variables (later updates of these initial values 
will be used for garbage collection).

\item \textit{\gchlock}: This is a global lock variable. The TM system
locks this variable whenever it wishes to read and write to any global variable.

\end{itemize}

The implementations of $T_i$'s {\op}s, denoted by $read_i(x)$, $write_i(x, v)$ and $\tryc_i()$ are described below. We assume here that if any of these is the first operation performed by $T_i$, it is preceded with the initialization all $T_i$'s local variables. 

We also assume that all the \tobj{s} accessed by the STM system are initialized with $0$ (which simulates the effect of having the initializing transaction $T_0$). 


\vspace{1mm}
\noindent
\textit{$read_i(x)$:} 
Every transaction $T_i$ maintains $\lchist_i$ which is a local copy
$\gchist$ combined with events of $T_i$ taken place so far, put at
the right places in $\gchist$, based on their sequence numbers. 
From $\lchist_i$ the values $v$ and $\lseq_i$ are computed. 
If there are no committed writes \op{} on $x$ preceding $read_i(x)$ in
$\lchist_i$, then $v$ is assumed to be the initial value $0$. 
Then, a \rvar{} $\rop_i$ is created for the current read \op{} using
the latest value of $x$, $v$ and the current value of \gseqn,
$\lseq_i$. Then $\rop_i$ is inserted into $\lchist_i$.    
A conflict graph is constructed from the resulting $\lchist_i$ and
checked for acyclicity. If the graph is cyclic then $A$ is inserted
into $\rop_i$ of $\lchist_i$ and then \texttt{abort} is returned. 
Otherwise, the value $v$ is returned. 

\vspace{1mm}
\noindent
\textit{$write_i(x, v)$:} adds a \wvar{} containing $x$ and $v$ and
$\lseq_i$ is inserted to $\lchist_i$.
(If the write is the first operation of $T_i$, then $\lchist_i$ and
$\lseq_i$ are computed based on the current state of  $\gchist_i$.)

\vspace{1mm}
\noindent
\textit{$\tryc_i(x)$:} The main idea for this procedure is similar to
$read_i$, except that the TM system first obtains the lock on
$\gchlock$. Then it makes local copies of $\gseqn$, $\gchist{}$ which are
$\lseq_i$, $\thist_i$, and $\lchist_i$. The value $\lseq_i$ is
incremented, and the $\cop_i(\lseq_i)$ item is appended to $\lchist_i$.
Then a conflict graph is constructed for the resulting $\lchist_i$ 
and checked for acyclicity. If the graph is cyclic then
$\cop_i(seq_i)$ is replaced with $a_i$ in $\lchist_i$, 
the lock is released and \texttt{abort} is returned. 
Otherwise,  $\lseq_i, \lchist_i$, are copied back into $\gseqn$,
$\gchist{}$, the lock is released and $ok$ is returned. \\


\subsection{Correctness of \sgt{}}
\label{subsec:proof}

In this section, we will prove 
that our implementation is permissive w.r.t. \svsls. 
Consider the history $H$ generated by \sgt{} algorithm. 
Recall that only read, \tryc{} and write \op{} (if it is the
first operation in a transaction) access shared memory. 
Hence, we call such operations  \emph{\memop{s}}. 

Note that $H$ is not necessarily sequential:
the transactional \op{s} can execute in overlapping manner. 
Therefore, to reason about correctness, we first
order all the \op{s} in $H$ to get an equivalent sequential
history. We then show that this sequential history is permissive with
respect to \clo.

We place the \memop{s}, say $r_i(x, v/A), \tryc_j(C/A)$ based on the
order in which they access  the global variable \gchist, storing the
history of currently committed transactions. 
The remaining write \op{s} are placed anywhere 
between the last preceding \memop{} and its $\tryc_i$ \op. 
We denote the resulting history, completed by adding 
$\tryc_i(A)$ \op{} for every incomplete transaction $T_i$, 
by $H_g$. 
It can be seen that $H_g$ represents a complete sequential
history that 
respects the real time ordering of memory operations in $H$.
In the rest of this section, we show that $H_g$ is permissive
(and, thus, non-interfering) with respect to \clo. 

\cmnt{
\begin{lemma}
\label{lem:gen-seq}
The history $H_g$ represents the correct ordering of the \op{s} of \sgt{} algorithm.
\end{lemma}
\begin{proof}
To prove this, let us analyse the \op{s} one by one. A read \op{} $r_i(x,v)$, accesses the one global variable \gchist. It reads this variable atomically. Thus any two read \op{s} can be ordered based on their access to \gchist. A \tryc{} \op{} obtains the lock \gchlock{} to perform read and write onto \gchist. Suppose two \tryc{} \op{s}  $\tryc_i$,  $\tryc_j$ executing simultaneously, try to lock \gchlock. Then only one of them, say $t_i$ will succeed and access \gchist. Then the other transaction $t_j$ can access \gchist{} only after $t_i$ releases the lock. Similarly the \memop{s} read and \tryc{} can be ordered. 

The writes of transactions are ordered after all the reads but before \tryc. The writes of different transactions execute only on the local memory. Hence they can be ordered in any way as long they respect the above mentioned transaction order. \qed
\end{proof}
}

Since \clo{} is local, to show that $H_g$ is in \clo,  
it is sufficient to show that, for each transaction $T_i$ in $\txns(H_g)$, 
$\sub{T_i}{H_g}$ is in \clo. 
We denote $\sub{T_i}{H_g}$ by $H_{ig}$. 

Consider a transaction $T_i \in \txns(H_g)$. 
Consider the last complete \memop{} of $T_i$ in $H$, denoted as
$m_i$. 
Note that every $T_i$ performs at least
one successful \memop{} (the proof for the remaining case is trivial). 
We define a history $H_{im}$ as the local history $\lchist_i$ computed
by \sgt{} with the last complete \memop{} of $T_i$ in $H$
(line~\ref{lin:lchist1} of Algorithm~\ref{algo:read} and
line~\ref{lin:lchist2} of Algorithm~\ref{algo:tryc}). 

\cmnt {
Consider a transaction $T_i \in \txns(H_g)$. We define a history $H_{im}$ as 
the local history $\lchist_i$ computed by \sgt{} in the last complete \memop{} of $T_i$
in $H$, let it be denoted $m_i$, that did not return \texttt{abort}
(line~\ref{lin:lchist1} of Algorithm~\ref{algo:read})
and line~\ref{lin:lchist2} of Algorithm~\ref{algo:tryc}).
By the algorithm, the history corresponds to an acyclic conflict graph
$CG(H_{im})$ (cf. checks in line~\ref{acycl:check1} of Algorithm~\ref{algo:read}
and line~\ref{acycl:check2} of Algorithm~\ref{algo:tryc}).
}


\begin{lemma}
\label{lem:equiv}
$H_{im}$ and $H_{ig}$ are equivalent.
\end{lemma}
\begin{proof}
Obviously, $H_{im}$ and $H_{ig}$ agree on the events of $T_i$.   
The \sgt{} algorithm assigns \cseq{} (a sequence number) to each
committed transaction $T_j$. 
Similarly it also assigns \rseq{} to each successfully completed read
\op, i.e. the read that did not return \texttt{abort}. 
Based on these sequence numbers, the \sgt{} algorithm constructs
$H_{im}$ (line~\ref{lin:lchist1} of Algorithm~\ref{algo:read},
and line~\ref{lin:lchist2}  of Algorithm~\ref{algo:tryc})
of all the events that committed before the last successful \memop{} of $T_i$ in $H_g$.
On the other hand, every event that appears in $H_{im}$ belongs to $T_i$ 
or a transaction that committed before the last successful \memop{} of $T_i$ in $H_g$.
Thus, $H_{im}$ and $H_{ig}$ are equivalent. \qed
\end{proof}

Even though $H_{im}$ and $H_{g}$ are equivalent, the ordering of the
events in these histories could be different. 
However, the two histories agree on the real-time and conflict orders
of transactions.
\begin{lemma}
\label{lem:orders}
$\prec_{H_{im}}^{CO}=\prec_{H_{ig}}^{CO}$ and $\prec_{H_{ig}}^{RT}=\prec_{H_{im}}^{RT}$
\end{lemma}
\begin{proof}
We go case by case for each possible relation in $\prec^{CO}\cup\prec^{RT}$.

\noindent
\emph{Write-write order:} we want to show that 
$(\tryc_p <_{im} \tryc_q) \Leftrightarrow (T_p.\cseq < T_q.\cseq) \Leftrightarrow (\tryc_p <_{ig} \tryc_q)$.

The result $(\tryc_p <_{im} \tryc_q) \Leftrightarrow (T_p.\cseq < T_q.\cseq)$ follows from the construction of $H_{im}$. We have already shown earlier that \tryc{} \op{} is atomic. When a transaction $T_i$ successfully commits in the \sgt{} algorithm, it is assigned an unique \cseq{} which is monotonically increasing. As a result, a \tryc{} \op{} which commits later gets higher \cseq{} in $H_g$. Since the ordering of events in $H_g$ are same as $H_{ig}$, we get that $(T_p.\cseq < T_q.\cseq) \Leftrightarrow (\tryc_p <_{ig} \tryc_q)$. 

\noindent
\emph{Write-read order:} For a committed transactions $T_p$ and a
successful read \op{} 
$r_q$, we want to show that 
$(\tryc_p <_{im} r_q) \Leftrightarrow (T_p.\cseq \leq r_q.\rseq) \Leftrightarrow (\tryc_p <_{ig} r_q)$.

The result $(\tryc_p <_{im} r_q) \Leftrightarrow (T_p.\cseq \leq
r_q.\rseq)$ follows from 
the construction of $H_{im}$. The \sgt{} algorithm stores as a part of
the read \op{} $r_j$,  \rseq{}  which is same as the \cseq{} of the
latest transaction that  committed before $r_j$, say $T_i$. Thus
$T_i.\cseq = r_j.\rseq$.  From the above argument for the write-write
order, we have that  any transaction $T_k$ that committed before $T_i$
will have lower  \cseq. This holds in $H_g$ and as a result also holds
in $H_{ig}$.  This shows that $(T_p.\cseq \leq r_q.\rseq) \Leftrightarrow (\tryc_p <_{ig} r_q)$.

\noindent
\emph{Read-write order:} For a committed transactions $T_q$ and a successful read \op{} $r_p$, we want to show that $(r_p <_{im} \tryc_q) \Leftrightarrow (r_p.\rseq < T_q.\cseq) \Leftrightarrow (r_p <_{ig} \tryc_q)$. The reasoning is similar to the above cases. 

\noindent Hence, $\prec_{H_{im}}^{CO}=\prec_{H_{ig}}^{CO}$.

\noindent
\emph{Real-time order:} Consider two transaction $T_p$, $T_q$ in
$H_{ig}$ such that $T_p \prec_{H_{ig}}^{RT} T_q$ which also holds in
$H_g$.  From the construction of $H_{ig}$, we get that $T_p$ is a
committed transaction  with its last event being $\tryc_p$. Indeed,
the only possibly uncommitted transaction in $H_{ig}$ is $T_i$ that
performs the last event in $H_{ig}$ and, thus, cannot precede any
transaction in $\prec_{H_{ig}}^{RT}$. 

Consider the first \memop{} of $T_q$ (by our assumption, there is one
in each $T_q$ in $H_{ig}$). 
By the algorithm, the sequence number associated with the \memop{} is 
at least as high as the sequence number of $\tryc_p$.
Thus $T_p\prec_{H_{im}}^{RT} T_q$ 
The other direction is analogous.

Hence, $\prec_{H_{im}}^{RT}=\prec_{H_{ig}}^{RT}$. \qed
\end{proof}

Lemmas~\ref{lem:equiv} and~\ref{lem:orders} imply that $H_{im}$ 
and $H_{ig}$ generate the same conflict graph:
\begin{corollary}
\label{cor:graphs-ia-ig}
$\cg{H_{ig}}=\cg{H_{im}}$
\end{corollary}

\noindent Now we argue about \svvalidity{} of $H_{im}$ and $H_{ig}$.
\begin{lemma}
\label{lem:ig-svvalid}
$H_{ig}$ is \svvalid.
\end{lemma}
\begin{proof}
By the algorithm, every successful read \op{} on a variable $x$ within
$T_i$ returns the argument of the last committed write on $x$ 
that appears in $\lchist_i$ (and, thus, in $H_{im}$).   
By applying this argument to every prefix of $H_{im}$, 
we derive that $H_{im}$ is \svvalid. By Lemmas~\ref{lem:ap-eqv-legal} and~\ref{lem:orders}, 
we derive that $H_{ig}$ is also legal. \qed
\end{proof}

\begin{theorem}
\label{thm:gen-clo}
Let $H_g$ be a history generated by the \sgt{} algorithm. 
Then $H_g$ is in \clo.
\end{theorem}
\begin{proof}
By the algorithm, the corresponding $H_{im}$ produces an acyclic conflict graph 
$CG(H_{im})$ (cf. checks in line~\ref{acycl:check1} of
Algorithm~\ref{algo:read} 
and line~\ref{acycl:check2} of Algorithm~\ref{algo:tryc}).
By~\corref{graphs-ia-ig}, $CG(H_{ig})$ is also acyclic.

Thus, by Theorem~\ref{thm:graph} and Lemma~\ref{lem:ig-svvalid}, 
for every $T_i\in\txns(H_g)$, $H_{ig}$
is \coop. Since \clo{} is a local property, we derive that $H_g$ is in
\clo. \qed
\end{proof}

Having proved that \sgt{} algorithm generates \clo{} histories, 
we now show that \sgt{} algorithm is in fact permissive w.r.t. \clo{}. 

\begin{theorem}
\label{thm:perm}
Let $H_g$ be a history generated by \sgt{} algorithm. 
Then $H_g$ is in $\permfn{\clo}$. 
\end{theorem}
\begin{proof}
We shall prove this by contradiction. Assume that $H_g$ is not in
$\permfn{\clo}$. From Theorem~\ref{thm:gen-clo}, we get that $H_g$ is in \clo.
Hence, condition (2) of \defref{perm} is not true. 
Thus, there is an aborted transaction $T_a$ in $H_g$ which can be
committed so that the resulting history is still in \clo. We denote the
modified transaction as $T_a^C$ and the resulting history as $H'_g$. 
There are two cases depending on the final event of $T_a$: 

\noindent
\textit{Case 1: The last event of $T_a$ is a read \op{} $r_a(x,A)$.} 
In order for $T_a^C$ to be committed in $H'_g$, $r_a(x,A)$ is
converted  to $r_a(x,v)$ for some $v$. If $H'_g$ is in \clo, then  
$\subs{T_a^C}{H'_g}$ is \coop. By~\corref{coop-legal}, we get that
$\subs{T_a^C}{H'_g}$ 
is legal. 
Therefore, $v$ is the value written by the transaction committing
$r_a$'s  \lastw{} in $H'_g$ (the current value on $v$). It can be seen
that $H'_g$ differs from $H_g$ only in $r_a$. 

But when \sgt{} algorithm attempts to read this value of $x$ in 
line~\ref{lin:lchist1} of Algorithm~\ref{algo:read}, it causes the conflict graph
maintained to be cyclic. From \corref{graphs-ia-ig} applied to $H'_g$,
we get that the conflict graph of $\subs{T_a^C}{H'_g}$ is also cyclic. 
By~\thmref{graph}, we derive that $\subs{T_a^C}{H'_g}$ is not
\coop. This implies that $H'_g$ is not in \clo{}---a contradiction.

\noindent
\textit{Case 2: The last event of $T_a$ is an abort \op{} $\tryc_a(A)$.}  
The argument in this case is similar to the above
case.  In order for $T_a^C$ to be committed in $H'_g$, $\tryc_a(A)$ is
converted into $\tryc_a(C)$. When \sgt{} algorithm attempts to commit
$T_a$ in line~\ref{lin:lchist2} of Algorithm~\ref{algo:tryc}, it causes the conflict graph
maintained to be cyclic.
By~\corref{graphs-ia-ig} applied to $H'_g$,
we derive that the conflict graph of $\subs{T_a^C}{H'_g}$ is also cyclic.
From \thmref{graph}, we then get that $\subs{T_a^C}{H'_g}$ is not
\coop.  This implies that $H'_g$ is not in \clo{} and hence again a
contradiction.

Therefore, no transaction $T_a$ in $H_g$ can not be transformed into a
committed  transaction $T_a^C$ while still staying in \clo. 
Hence, $H_g$ is in $\permfn{\clo}$. \qed
\end{proof}
It is left to show that our algorithm is \emph{live}, i.e., under certain
conditions, every operation eventually completes.

\begin{theorem}
\label{thm:live}
Assuming that no transaction fails while executing the \tryc{} \op{}
and $\gchlock$ is starvation-free, every operation of \sgt{} eventually returns. 
\end{theorem}

\begin{proof}
It can be seen that \textit{read} and \textit{write} functions do not involve
any waiting. 
Therefore, $\tryc$ is the only function which involves waiting for 
the $\gchlock$ variable. But since the lock is starvation-free
and no transaction executing $\tryc$ obtains the lock forever, 
every such waiting is finite. 
Thus, every $\tryc$ \op{} eventually grabs the lock and, after,
computing the outcome, returns. \qed
\end{proof}

\begin{theorem}
\label{thm:main-gen-clo}
Let $H_g$ be a history generated by the \sgt{} algorithm. Then $H_g$ is in \clo.
\end{theorem}

\begin{theorem}
\label{thm:main-perm}
Let $H_g$ be a history generated by \sgt{} algorithm. Then $H_g$ is in $\permfn{\clo}$. 
\end{theorem}
Now \thmref{perm_ni} implies that our \sgt{} implementation is \clo-non-interfering.

\begin{theorem}
\label{thm:main-live}
Assuming that no transaction fails while executing the \tryc{} \op{} and $\gchlock$ is starvation-free, every operation of \sgt{} eventually returns. 
\end{theorem}

\cmnt {
\vspace{1mm}
\noindent
\textit{Garbage Collection.} Over time, the history of committed transactions maintained by our \sgt{} algorithm in the global variable  $\gchist$ grows without bound. To check the size of \gchist, we have developed a simple garbage-collection scheme that allows to keep the size of $\gchist$ proportional to the current contention, i.e, to the number of concurrently live transactions. The procedure has been described in \cite{KP12:TR}. The idea is to periodically remove from $\gchist$ the sub-histories corresponding to committed transactions that become \emph{obsolete}, i.e., the effect of them can be reduced to the updates of t-objects.
}

\subsection{Garbage Collection}
\label{subsec:garbage}

\cmnt {
\newcommand{\liveset} {\id{liveSet}}
\newcommand{\primary} {primary}
\newcommand{\glset} {gIncSet}
\newcommand{\waits} {waitSet}
\newcommand{\dset} {dependSet}
\newcommand{\sstate} {system-state}
}

Over time, the history of committed transactions maintained by our \sgt{}
algorithm in the global variable  $\gchist$ grows
without bound. We now describe a simple garbage-collection
scheme that allows to keep the size of $\gchist$ proportional to the
current contention, i.e, to the number of concurrently live transactions.
The idea is to periodically remove from $\gchist$ the sub-histories
corresponding to committed transactions that become \emph{obsolete}, i.e.,
the effect of them can be reduced to the updates of t-objects.

More precisely, a transaction $T_i$'s \emph{\liveset} is the set of
the transactions that were incomplete when $T_i$ terminated. 
A t-complete transaction $T_i$ is said to be \obsolete{} (in a history $H$) if all the transactions 
in its \liveset{} have terminated (in $H$).

To make sure that obsolete transactions can be correctly identified
based on the global history {\gchist}, we update our algorithm as
follows. When a transaction performs its first operation, it grabs the
lock on {\gchist} and inserts the operation in it. 
%
Now when a transaction commits it takes care of all committed
transactions in \gchist{} which have become
\obsolete. 
All read operations preceding the last event of an  {\obsolete}
transaction are removed,   
In case there are multiple \obsolete{} transactions writing to the
same t-object, only the writes of the last such {\obsolete}
transaction to commit are kept in the history.
If an {\obsolete} transaction is not the latest to commit an update on
any t-object, all events of this transactions are removed.   

In other words, $H_{im}$ defined as the local history $\lchist_i$ computed
by \sgt{} within the last complete \memop{} of $T_i$ in the updated algorithm
(which corresponds to line~\ref{lin:lchist1} of Algorithm~\ref{algo:read} and
line~\ref{lin:lchist2} of Algorithm~\ref{algo:tryc}) preserves write
and commit events of the  latest {\obsolete} transaction to commit a
value for every t-object. All other events of other {\obsolete} transactions
are removed.
The computed history $H_{im}$ is written back to  $\gchist$ in line~\ref{lin:update} of Algorithm~\ref{algo:read}.   

Let this $\gchist$ be used by a transaction $T_i$ in checking the
correctness of the current local history (line~\ref{acycl:check1} of
Algorithm~\ref{algo:read} or line~\ref{acycl:check2} of
Algorithm~\ref{algo:tryc}). 
Recall that $H_{ig}$ denotes the corresponding local history of $T_i$.
Let $T_{\ell}$ be any obsolete transaction in $H_{ig}$.
Note that all transactions that committed before $T_{\ell}$ in
$H_{ig}$ are also obsolete in $H_{ig}$, and let $U$ denote the
set of all these obsolete transactions, including $T_{\ell}$.
Respectively, let $\id{obs}(H_{ig},U)$ be a prefix of $H_{ig}$
in which all transactions in $\id{\liveset}(T_{\ell})$ are  complete.
Also, let $\id{trim}(H_{ig},U)$ be the ``trimmed'' local history of $T_i$
where all transactions in $U$ are removed or replaced with committed
updates, as described above.
We can show that $H_{ig}$ is in {\clo} if and only if $\id{obs}(H_{ig},U)$ and $\id{trim}(H_{ig},U)$ are in {\clo}.

Iteratively, for each $T_i$, all our earlier claims on the relation
between the actual local history $H_{ig}$  and the locally constructed
history $H_{im}$ 
hold now for the ``trimmed'' history $\id{trim}(H_{ig},U)$ and $H_{im}$.
Therefore, 
$H_{im}$ is in {\clo} if and only if
$H_{ig}$ is in {\clo}. Hence, every history $H_g$ generated
by the updated algorithm with garbage collection is \clo-permissive (and,
thus, \clo-non-interfering).

Note that removing obsolete transactions from  $\gchist$ essentially boils down to
dropping a prefix of it that is not concurrent to any live
transactions.
As a result, the length of $\gchist$ is $O(M+C)$, where $M$ is the
number of t-objects and $C$ is the
upper bound on the number of concurrent transactions.
A complete correctness proof for the optimized algorithm is given in~\cite{KP12:TR}.   

\section{Concluding remarks}
\label{sec:conc}

In this paper, we explored the notion of
non-interference in transactional memory, 
originally highlighted in~\cite{SatVid:2011:ICDCN, SatVid:2012:ICDCN}.
We focused on  $P$-non-interference that 
grasps the intuition that no transaction aborts because of aborted or incomplete
transactions in the sense that by removing some of aborted or incomplete
transactions we cannot turn a previously aborted transaction into a
committed one 
without violating the given correctness criterion $P$.
We showed that no TM implementation can provide
opacity-non-interference.
However, we observed that any permissive implementation of a local 
correctness criterion is also non-interfering.
Informally, showing that a history is locally correct is equivalent
to showing that every its local sub-history is correct.
We discussed two local criteria: virtual-world
consistency (VWC)~\cite{ImbsRay:2009:SIROCCO} and the (novel) local
opacity (LO).
Unlike VWC, LO does not allow a transaction that is
doomed to abort to waste system resources.
TMS1~\cite{DGLM13} was recently proposed as a candidate for the 
``weakest reasonable''  TM correctness criterion. 
Interestingly, at least for the case of atomic transactional
operations, LO seems to coincide with TMS1.   

We then considered \clo{}, a restriction of LO that, in addition, requires every
local serialization to respect the conflict
order~\cite{Papad:1979:JACM,Hadzilacos:1988:JACM} of the original sub-history. 
We presented a permissive, and thus non-interfering, \clo{} implementation.  
This appears to be the only non-trivial permissive implementation known so far 
(the VWC implementation in~\cite{CIR11} is only probabilistically permissive).  

Our definitions and our implementation intend to build a ``proof of concept''
for non-interference and are, by intention, as simple as possible (but not
simpler). Of course, interesting directions are to consider a more
realistic notion of non-interference as a characteristics of an
implementation, to extend our
definitions to non-sequential histories, and to
relax the strong ordering requirements in our correctness criteria.
Indeed, the use of the conflict order allowed us to efficiently relate 
correctness  of a given history to the absence of cycles in its graph
characterization. Respecting conflict order makes a lot of sense if we aim at
permissiveness, as efficient verification of strict
serializability or opacity appear elusive~\cite{Papad:1979:JACM}.
But it may be too strong as a requirement for less demanding implementations.

Also, our implementation is quite simplistic in the sense that 
it uses one global lock to protect the history of
committed transactions and, thus, it is not
disjoint-access-parallel (DAP)~\cite{israeli-disjoint,AHM09}.
An interesting challenge is to check if it is possible to construct a permissive
DAP \clo{} implementation with invisible reads. 

\bibliography{citations}

\begin{thebibliography}{10}

\bibitem{AHM09}
H.~Attiya, E.~Hillel, and A.~Milani.
\newblock Inherent limitations on disjoint-access parallel implementations of
  transactional memory.
\newblock In {\em Proceedings of the twenty-first annual symposium on
  Parallelism in algorithms and architectures}, SPAA '09, pages 69--78, New
  York, NY, USA, 2009. ACM.

\bibitem{CIR11}
T.~Crain, D.~Imbs, and M.~Raynal.
\newblock Read invisibility, virtual world consistency and probabilistic
  permissiveness are compatible.
\newblock In {\em ICA3PP (1)}, pages 244--257, 2011.

\bibitem{norec}
L.~Dalessandro, M.~F. Spear, and M.~L. Scott.
\newblock Norec: streamlining stm by abolishing ownership records.
\newblock In {\em PPOPP}, pages 67--78, 2010.

\bibitem{tinystm}
P.~Felber, C.~Fetzer, P.~Marlier, and T.~Riegel.
\newblock Time-based software transactional memory.
\newblock {\em IEEE Trans. Parallel Distrib. Syst.}, 21(12):1793--1807, 2010.

\bibitem{Guer+:disc:2008}
R.~Guerraoui, T.~Henzinger, and V.~Singh.
\newblock Permissiveness in transactional memories.
\newblock In {\em DISC~'08: Proc. 22nd International Symposium on Distributed
  Computing}, pages 305--319, sep 2008.
\newblock Springer-Verlag Lecture Notes in Computer Science volume 5218.

\bibitem{GuerKap:2008:PPoPP}
R.~Guerraoui and M.~Kapalka.
\newblock On the correctness of transactional memory.
\newblock In {\em PPoPP '08: Proceedings of the 13th ACM SIGPLAN Symposium on
  Principles and practice of parallel programming}, pages 175--184, New York,
  NY, USA, 2008. ACM.

\bibitem{tm-book}
R.~Guerraoui and M.~Kapalka.
\newblock {\em Principles of Transactional Memory,Synthesis Lectures on
  Distributed Computing Theory}.
\newblock Morgan and Claypool, 2010.

\bibitem{Hadzilacos:1988:JACM}
V.~Hadzilacos.
\newblock A theory of reliability in database systems.
\newblock {\em J. ACM}, 35(1):121--145, Jan. 1988.

\bibitem{ImbsRay:2009:SIROCCO}
D.~Imbs and M.~Raynal.
\newblock A versatile {STM} protocol with invisible read operations that
  satisfies the virtual world consistency condition.
\newblock In {\em Proceedings of the 16th international conference on
  Structural Information and Communication Complexity}, SIROCCO'09, pages
  266--280, Berlin, Heidelberg, 2010. Springer-Verlag.

\bibitem{israeli-disjoint}
A.~Israeli and L.~Rappoport.
\newblock Disjoint-access-parallel implementations of strong shared memory
  primitives.
\newblock In {\em Proceedings of the thirteenth annual ACM symposium on
  Principles of distributed computing}, PODC '94, pages 151--160, New York, NY,
  USA, 1994. ACM.

\bibitem{KP12:TR}
P.~Kuznetsov and S.~Peri.
\newblock Non-interference and locality in transactional memory.
\newblock {\em CoRR}, abs/1211.6315, 2012.

\bibitem{KR:2011:OPODIS}
P.~Kuznetsov and S.~Ravi.
\newblock On the cost of concurrency in transactional memory.
\newblock In {\em OPODIS}, pages 112--127, 2011.

\bibitem{Papad:1979:JACM}
C.~H. Papadimitriou.
\newblock The serializability of concurrent database updates.
\newblock {\em J. ACM}, 26(4):631--653, 1979.

\bibitem{SatVid:2011:ICDCN}
S.~Peri and K.Vidyasankar.
\newblock Correctness of concurrent executions of closed nested transactions in
  transactional memory systems.
\newblock In {\em 12th International Conference on Distributed Computing and
  Networking}, pages 95--106, 2011.

\bibitem{SatVid:2012:ICDCN}
S.~Peri and K.Vidyasankar.
\newblock An efficient scheduler for closed nested transactions that satisfies
  all-read-consistency and non-interference.
\newblock In {\em 13th International Conference on Distributed Computing and
  Networking}, 2012.

\bibitem{spear+:spaa:ringstm:2008}
M.~F. Spear, M.~M. Michael, and C.~von Praun.
\newblock Ringstm: scalable transactions with a single atomic instruction.
\newblock In {\em Proceedings of the twentieth annual symposium on Parallelism
  in algorithms and architectures}, SPAA '08, pages 275--284, 2008.

\bibitem{WeiVoss:2002:Morg}
G.~Weikum and G.~Vossen.
\newblock {\em Transactional Information Systems: Theory, Algorithms, and the
  Practice of Concurrency Control and Recovery}.
\newblock Morgan Kaufmann, 2002.

\end{thebibliography}

\appendix


\section{Appendix}

\subsection{Graph characterization of \coopty}
\label{subsec:ap-graph}

In the following lemmas, we show that the graph characterization indeed helps us verify the membership in \coopty. Note, since $\txns(H)=\txns(\overline{H})$ and ($\prec_H^{RT}\cup\prec_H^{CO})=(\prec_{\overline{H}}^{RT}\cup\prec_{\overline{H}}^{CO}$), we have $\cg{H} = \cg{\overline{H}}$. 

\begin{lemma}
\label{lem:ap-co-eq}
Consider two histories $H1$ and $H2$ such that $H1$ is equivalent to
$H2$ and $H1$ respects conflict order of $H2$, i.e., $\prec_{H1}^{CO} \subseteq \prec_{H2}^{CO}$. Then, $\prec_{H1}^{CO} = \prec_{H2}^{CO}$. 
\end{lemma}

\begin{proof}
Here, we have that $\prec_{H1}^{CO} \subseteq \prec_{H2}^{CO}$. In order to prove $\prec_{H1}^{CO} = \prec_{H2}^{CO}$, we have to show that $\prec_{H2}^{CO} \subseteq \prec_{H1}^{CO}$. We prove this using contradiction. Consider two events $p,q$ belonging to transaction $T1,T2$ respectively in $H2$ such that $(p,q) \in \prec_{H2}^{CO}$ but $(p,q) \notin \prec_{H1}^{CO}$. Since the events of $H2$ and $H1$ are same, these events are also in $H1$. This implies that the events $p, q$ are also related by $CO$ in $H1$. Thus, we have that either $(p,q) \in \prec_{H1}^{CO}$  or $(q,p) \in \prec_{H1}^{CO}$. But from our assumption, we get that the former is not possible. Hence, we get that $(q,p) \in \prec_{H1}^{CO} \Rightarrow (q,p) \in \prec_{H2}^{CO}$. But we already have that $(p,q) \in \prec_{H2}^{CO}$. This is a contradiction. \qed
\end{proof}

\begin{lemma}
\label{lem:ap-eqv-legal}
Let $H1$ and $H2$ be equivalent histories such that 
$\prec_{H1}^{CO} = \prec_{H2}^{CO}$. 
Then $H1$ is \legal{} iff $H2$ is \legal. 
\end{lemma}

\begin{proof}
It is enough to prove the `if' case, and the `only if' case will follow
from symmetry of the argument. 
Suppose that $H1$ is \legal{}. 
By contradiction, assume that $H2$ is not \legal, i.e.,
there is a read \op{} $r_j(x,v)$ (of transaction $T_j$) in $H2$ with
\lastw{} as $c_k$ (of transaction $T_k$) and $T_k$ writes $u \neq v$
to $x$, i.e $w_k(x, u) \in \evts{T_k}$. 
Let $r_j(x,v)$'s \lastw{} in $H1$ be $c_i$ of $T_i$. 
Since $H1$ is legal, $T_i$ writes $v$ to $x$, i.e $w_i(x, v) \in
\evts{T_i}$. 

Since $\evts{H1} = \evts{H2}$, we get that $c_i$ is also in $H2$, and
$c_k$ is also in $H1$. 
As $\prec_{H1}^{CO} = \prec_{H2}^{CO}$,
we get $c_i <_{H2} r_j(x, v)$ and $c_k <_{H1} r_j(x, v)$. 

Since $c_i$ is the \lastw{} of $r_j(x,v)$ in $H1$ we derive that 
$c_k <_{H1} c_i$ and, thus, $c_k <_{H2} c_i <_{H2} r_j(x, v)$.
But this contradicts the assumption that $c_k$ is the \lastw{} of
$r_j(x,v)$ in $H2$.
Hence, $H2$ is legal. \qed
\cmnt{
from the w-r conflict order we get that 
\begin{equation}
\label{eq:wr}
c_i <_{H2} r_j(x, v)
\end{equation}

Now, we have two cases based on the ordering of $c_k$ and $r_j$ in $H1$:

\begin{itemize}
\item[] Case (1) $r_j(x, v) <_{H1} c_k$: Since $\prec_{H1}^{CO} = \prec_{H2}^{CO}$, from r-w conflict order we get that $r_j <_{H2} c_k$. This implies that $c_k$ can not be $r_j$'s \lastw{} in $H2$ which is a contradiction. 

\item[] Case (2) $c_k <_{H1} r_j(x, v)$: From w-r conflict order equivalence of $\prec_{H1}^{CO} = \prec_{H2}^{CO}$, we get that $c_k <_{H2} r_j$. Here, we again have two cases based on the ordering of $c_i$ and $c_k$ in $H1$, Case (2.1) $c_i <_{H1} c_k$: In this case, we have that $c_i <_{H1} c_k <_{H1} r_j$ which implies that $c_i$ is not $r_j$'s \lastw{} in $H2$, a contradiction. Case (2.2) $c_k <_{H1} c_i$: From w-w conflict order equivalence of $\prec_{H1}^{CO} = \prec_{H2}^{CO}$, we get that $c_k <_{H2} c_i$. Combining this with \eqnref{wr}, we have that $c_k <_{H2} c_i <_{H2} r_j$. This implies that $c_k$ can not be $r_j$'s \lastw{} in $H2$ which is a contradiction.
\end{itemize}
Thus in all the cases, we get that $c_k$ can not be $r_j$'s
\lastw. This implies that $H2$ is legal. 
}
\end{proof}
From the above lemma we get the following interesting corollary.

\begin{corollary}
\label{cor:coop-legal}
Every \coop{} history $H$ is \legal{} as well.
\end{corollary}
Based on the conflict graph construction,
we have the following graph characterization for \coop.

\begin{theorem}
\label{thm:ap-graph}
A \legal{} history $H$ is \coop{} iff $\cg{H}$ is acyclic. 
\end{theorem}

\begin{proof}

\noindent
\textit{(Only if)} If $H$ is \coop{} and legal, then
$\cg{H}$ is acyclic:
Since $H$ is \coop{}, there exists a legal t-sequential
history $S$ equivalent to $\overline{H}$ and $S$ respects
$\prec_{H}^{RT}$ and $\prec_{H}^{CO}$. Thus from the conflict graph
construction we have that $\cg{\overline{H}} (= \cg{H})$ is a sub
graph of $\cg{S}$. Since $S$ is sequential, it can be inferred that
$\cg{S}$ is acyclic. Any sub graph of an acyclic graph is also
acyclic. Hence $\cg{H}$ is also acyclic. 

\vspace{1mm}
\noindent
\textit{(if)} If $H$ is \legal{} and $\cg{H}$ is acyclic then $H$ is
\coop: 
Suppose that $\cg{H}=\cg{\overline{H}}$ is acyclic. Thus we can
perform a topological sort on the vertices of the graph and obtain a
sequential order. Using this order, we can obtain a sequential
schedule $S$ that is equivalent to $\overline{H}$.
Moreover, by construction, $S$ respects $\prec_{H}^{RT} =
\prec_{\overline{H}}^{RT}$ 
and $\prec_{H}^{CO} = \prec_{\overline{H}}^{CO}$. 

Since every two events related by the conflict relation (w-w, r-w, or
w-r)in $S$ are also related by $\prec_{\overline{H}}^{CO}$, we obtain 
$\prec_{S}^{CO} = \prec_{\overline{H}}^{CO}$. 
Since $H$ is legal, $\overline{H}$ is also legal. 
Combining this with \lemref{ap-eqv-legal}, we get that $S$ is also
\legal. 
This satisfies all the conditions necessary for $H$ to be \coop. \qed
\end{proof}

\subsection{Proofs of local correctness and non-interference}
\label{subsec:ap-local}

\begin{theorem}
\label{lem:ap-perm_ni}
For every local correctness property $P$, 
$\permfn{P}\subseteq \nifn{P}$.
\end{theorem}

\begin{proof}
We proceed by contradiction. Assume that $H$ is in
$\permfn{P}$ but not in $\nifn{P}$. 
More precisely, let $T_a$ be an aborted transaction in $H$, 
$R \subseteq \prevA{T_a}{H}$ and  
$\widetilde H \in \H^{T_a,C}_{-R}$, such that  $\widetilde H \in P$.

On the other hand, since $H\in \permfn{P}$, we have  
$\H^{T_a,C} \cap P = \emptyset$.
Since $P$ is local and $H\in P$, we have $\forall T_i\in\txns(P)$, 
 $\subs{T_i}{H}\in P$. 
Thus, for all transactions $T_i$ that
 committed before the last event of $T_a$, we have  
$\subs{T_i}{H}=\subs{T_i}{H^{T_a}}\in P$.    

Now we construct $\widehat H$ as $H^{T_a}$, except that the aborted
operation of $T_a$ is replaced with the last operation of $T_a$ in
$\widetilde H$.  Since $\widetilde H$ is in $P$, and $P$ is local, 
we have $\subs{T_a}{\widehat H} =
\subs{T_a}{\widetilde H}\in P$. 
For all transactions $T_i$ that committed before the last event of $T_a$ in
$\widehat H$, we have  $\subs{T_i}{\widehat H}=\subs{T_i}{H^{T_a}}\in P$.
Hence, since $P$ is local, we have $\widehat H\in P$.
But, by construction, $\widehat H\in\H^{T_a,C}$---a contradiction with
the assumption that $\H^{T_a,C} \cap P = \emptyset$.
\qed
\end{proof}

\subsection{Proof for Garbage Collection}
\label{subsec:ap-garbage}

\begin{lemma}
\label{lem:ap-trimmed}   
$H_{ig}$ is in {\clo} if and only if $\id{obs}(H_{ig},U)$ and $\id{trim}(H_{ig},U)$ are in {\clo}. 
\end{lemma}

\begin{proof}
(Only if)
Suppose that $H_{ig}$ is in {\clo}.
By \corref{coop-legal}, $H_{ig}$ is legal.
Since $\id{obs}(H_{ig},U)$ is a prefix of $H_{ig}$, it is also legal,
and its conflict
graph is a sub-graph of $CG(H_{ig}$. By Theorem~\ref{thm:graph}, 
$CG(\id{obs}(H_{ig},U))$ is acyclic and, thus, $\id{obs}(H_{ig},U)$ is
in {\clo}.

Further, let $r_k(x,v)$ be any read operation in
$\id{trim}(H_{ig},U)$.
Since $H_{ig}$ is legal,  $r_k(x,v)$ is also legal.
Note that since no read operation of {\obsolete} transactions in
$H_{ig}$ appears in $\id{trim}(H_{ig},U)$, $T_k$ is not in
$U$.  
Let $c_m$ be $r_k(x,v)$ 's \textit{\lastw{}} in $H_{ig}$. 
If $c_m$ appears in  $\id{trim}(H_{ig},U)$, then 
$c_m$ is also $r_k(x,v)$ 's \textit{\lastw{}} in
$\id{trim}(H_{ig},U)$, and, thus, $r_k(x,v)$ is also legal.    
Now, suppose, by contradiction, that $c_m$ does not appear in
$\id{trim}(H_{ig},U)$, i.e., $c_m$ is not the last ({\obsolete})
transaction in $U$ to commit a value on $x$, i.e., there exists a
transaction $T_s\in U$ writing to $x$ such that $c_s$ appears after
$c_m$  in $H_{ig}$.
Since $c_m$ is $r_k(x,v)$ 's \textit{\lastw{}} in $H_{ig}$, $c_s$
appears after $r_k(x,v)$ in $H_{ig}$. 
But $T_s$ is {\obsolete}, and, thus, no read operation 
can appear before $c_s$ in $\id{trim}(H_{ig},U)$---a contradiction.   
Thus, $c_m$ is $r_k(x,v)$ 's \textit{\lastw{}} in
$\id{trim}(H_{ig},U)$, and, hence, $\id{trim}(H_{ig},U)$ is legal.

Since  $\id{trim}(H_{ig},U)$ is a legal sub-sequence of $H_{ig}$,
$CG(\id{trim}(H_{ig},U))$ is a sub-graph of $CG(H_{ig})$ and, by Theorem~\ref{thm:graph}, 
$CG(\id{trim}(H_{ig},U))$ is acyclic and $\id{trim}(H_{ig},U)$ is in {\clo}.

\noindent
(If) Suppose now that $\id{obs}(H_{ig},U)$ and $\id{trim}(H_{ig},U)$ are in {\clo}. 
By  \corref{coop-legal}, both histories are legal, and, by
Theorem~\ref{thm:graph}, produce acyclic conflict graphs.
Immediately, every read operation in $H_{ig}$ that also appears in
$\id{obs}(H_{ig},U)$ is legal.
By the arguments above, the {\lastw} for every read operation in
$\id{trim}(H_{ig},U)$ is also its {\lastw} in $H_{ig}$.
Thus, $H_{ig}$ is legal.

Recall that $H_{ig}$ can be represented as $\id{trim}(H_{ig},U)$ with
read events of transactions in $U$ inserted in accordance to its prefix  $\id{obs}(H_{ig},U)$.
Thus, 
$CG(H_{ig})$ can be represented as $CG(\id{trim}(H_{ig},U))$ 
with several additional edges directed to and from transactions in
$U$.

Suppose, by contradiction that $CG(H_{ig})$ contains a cycle $C$. 
Since $CG(\id{trim}(H_{ig},U))$ is acyclic, $C$ must contain an edge
directed to or from a transaction in $U$ that does not appear in $CG(\id{trim}(H_{ig},U))$.   

Thus, we can represent the cycle $C$ as
$T_{i_1},T_{i_2},\ldots,T_{i_k}$, where $T_{i_1}=T_{i_k}\in U$ and for all
$j=1,\ldots,k-1$, $(T_{i_j},T_{i_{j+1}})\in CG(H_{ig})$. 
Since  $CG(\id{obs}(H_{ig},U))$ is acyclic $C$ must contain an edge
that does not appear in $CG(\id{obs}(H_{ig},U))$. 
Let $T_{i_j}$ be the latest transaction in
$T_{i_2},\ldots,T_{i_k}$ such that $(T_{i_{j-1}},T_{i_j})\notin
CG(\id{obs}(H_{ig},U))$.

Note that $j\neq k$. This is because $T_{i_{k-1}}$
must precede or be concurrent to $T_{i_{k}}$ in $H_{ig}$. 
Since $T_{i_{k}}\in U$, by the construction of $\id{obs}(H_{ig},U)$,
$T_{i_{k-1}}$ must have committed in $\id{obs}(H_{ig},U)$.
But then $(T_{i_{j-1}},T_{i_j})\in CG(\id{obs}(H_{ig},U))$---a contradiction.

Now, since $(T_{i_{j-1}},T_{i_j})\notin
CG(\id{obs}(H_{ig},U))$, $T_{i_j}$ cannot be complete in
$\id{obs}(H_{ig},U)$.   
Again, by the construction of $\id{obs}(H_{ig},U)$, no transaction
that is not complete in $\id{obs}(H_{ig},U)$ can begin before
$T_{i_k}\in U$ completes. 
Hence, $T_{i_k}$ precedes $T_{i_j}$ in the real-time order and, since
$T_{i_j}$ also appears in $\id{obs}(H_{ig},U)$ , $(T_{i_{k}},T_{i_j})\in
CG(\id{obs}(H_{ig},U))$. 
Thus, $CG(\id{obs}(H_{ig},U))$ contains a cycle $T_{i_{k}},T_{i_j},
T_{i_{j+1}}\ldots,T_{i_k}$---a contradiction. 

Thus, $CG(H_{ig})$ is acyclic and, by Theorem~\ref{thm:graph}, $H_{ig}$
is in {\clo}. \qed
\end{proof}


%
\end{document}